\documentclass[12pt]{article}

\usepackage{graphicx}
\usepackage[dvips]{color}

\usepackage{amsmath,amssymb,epsf,epic}

\newcommand{\ra}{\rightarrow}
\newcommand{\tr }{\mbox{tr}}
\newcommand{\hol }{\mbox{Hol}}

\newcommand{\spr}{\mbox{\rm Spr}\,} 
 
\newcommand{\bra}{\langle} 

\newcommand{\ket}{\rangle}

\newcommand{\E}{{\mathbb E}}
\newcommand{\D}{{\mathbb D}}

\newcommand{\be}{\begin{equation}}
\newcommand{\ee}{\end{equation}}
\newcommand{\bea}{\begin{eqnarray}}
\newcommand{\eea}{\end{eqnarray}}

\newcommand{\ffi}{\varphi}
\newcommand{\sign}{\mbox{sign}}
\newcommand{\ep}{\hfill  {\vrule height 10pt width 8pt depth 0pt}}

\newcommand{\grintl}{[\kern-.18em [}
\newcommand{\grintr}{]\kern-.18em ]}

\newcounter{resultcounter}[section]

\newtheorem{thm}[resultcounter]{Theorem}
\newtheorem{lem}[resultcounter]{Lemma}
\newtheorem{prop}[resultcounter]{Proposition}
\newtheorem{cor}[resultcounter]{Corollary}
\newtheorem{definition}[resultcounter]{Definition}
\newtheorem{rem}[resultcounter]{Remark}
\newtheorem{rems}[resultcounter]{Remarks}

 \def\cB{{\cal B}} 
  
 \def\cH{{\cal H}}

\newcommand{\R}{{\mathbb R}}
\newcommand{\N}{{\mathbb N}}

\newcommand{\C}{{\mathbb C}}
\newcommand{\Z}{{\mathbb Z}}
\renewcommand{\E}{{\mathbb E}}
\renewcommand{\P}{{\mathbb P}}
\newcommand{\I}{{\mathbb I}}
\newcommand{\T}{{\mathbb T}}

\begin{document}
\title{Density of States for Random Contractions }
 
\author{Alain Joye\footnote{ UJF-Grenoble 1, CNRS Institut Fourier UMR 5582, Grenoble, 38402, France} }

\date{ }

\maketitle
\vspace{-1cm}

\thispagestyle{empty}
\setcounter{page}{1}
\setcounter{section}{1}

\setcounter{section}{0}

\abstract{
We define a linear functional, the DOS functional, on spaces of holomorphic functions on the unit disk which is associated with random ergodic contraction operators on a Hilbert space, in analogy with the density of state functional for random self-adjoint operators. The DOS functional is shown to enjoy natural integral representations on the unit circle and on the unit disk. For random contractions with suitable finite volume approximations, the DOS functional is proven to be the almost sure infinite volume limit of the trace per unit volume of functions of the finite volume restrictions. Finally, in case the normalised counting measure of the spectrum of the finite volume restrictions converges in the infinite volume limit, the DOS functional is shown admit an integral representation on the disk in terms of the limiting measure, despite the discrepancy between the spectra of non normal operators and their finite volume restrictions. Moreover, the integral representation of the DOS functional on the unit circle is related to the Borel transform of the limiting measure.
}

\thispagestyle{empty}
\setcounter{page}{1}
\setcounter{section}{1}

\setcounter{section}{0}

\section{Introduction}

The density of states measure is an important mathematical notion in the study of the spectral properties of random self-adjoint operators, with a well defined physical significance, see e.g. the textbooks \cite{cfks, cl, Ki}. It is the measure associated to a positive functional acting on compactly supported continuous functions on the real axis, related to the random self-adjoint operator. Given a function, the functional is defined as the expectation of a diagonal matrix element of the function of the random operator, and the Riesz representation theorem provides the associated  density of state measure.  For operators defined $\Z^d$, a physically appealing definition consists in considering the trace per unit volume of functions of the random operator restricted to boxes $\Lambda\subset \Z^d$ by suitable boundary conditions, and in taking the limit $\Lambda\ra \Z^d$. Under ergodicity assumptions, the limit exists almost surely and the two notions coincide. This procedure works equally well for unitary operators, see e.g. \cite{J1}.

We revisit these constructions in the framework of bounded random operators that are not necessary normal, and study some of their properties, as described below. By rescaling, we can restrict attention to contraction operators that we assume are completely non unitary (cnu for short). 

As the there is no continuous functional calculus in this framework, we resort to the holomorphic functional calculus for cnu contractions developed on a Hardy space of the disk, as recalled in Section \ref{fccnuc}. This allows us to define a density of state functional (DOS functional for short) on $H^\infty(\D)$ in analogy with that defined for self-adjoint operators, see Definition \ref{defdos}. We show  in Proposition \ref{repfun} that the DOS functional possesses a natural integral representation on the unit circle $\T$ by a function $\ffi\in L^1(\T)\setminus H_0^1(\D)$, and that when restricted to  the disc algebra $A(\D)=H^\infty(\D)\cup C(\bar\D)$, it further admits  an integral representation on the disk by a complex harmonic function $m_\ffi$, as Proposition \ref{harmden} shows. 

Restrictions of random contractions to finite volume boxes $\Lambda\subset \Z^d$ are considered in Section \ref{sfva}, under suitable ergodicity assumptions. The infinite volume limit of the trace per unit volume of functions of random contractions is shown to coincide with the DOS functional defined via the full operator in Propositions \ref{llt} and \ref{limth}. We make use of this alternative construction of the DOS functional on $A(\D)$ to show that 
{\it a priori} estimates of the spectral radius of the finite volume restrictions imply more structure on the integral representations on $\T$ and $\D$ by means of $\ffi$ and $m_\ffi$: these functions are shown to be defined by the complex conjugate of a  function $\psi$ that is holomorphic in a neighbourhood of the unit disk, see Theorem \ref{smoothfi} and Corollary \ref{scalprod}. Then, we consider the situation where the normalised counting measures on the spectra of finite volume restrictions of the random contractions converge weakly , in the infinite volume limit. Theorem \ref{convcount} states that, despite the fact that finite volume restrictions of non normal operators have spectra that generally differ significantly from the full operator \cite{D1, D2, GK2, TE}, the limiting measure provides us with yet an alternative integral representation of the DOS functional on the unit disk.
Moreover, the holomorphic function $\psi$ is directly related to the Borel transform of the limiting measure.

Some examples are worked out in Section \ref{illustr} to illustrate the various features of the DOS functional. 
We start with random contractions defined as multiples of random unitary operators. Then
we consider the non self-adjoint Anderson model (NSA model), whose finite volume restrictions have eigenvalue distributions that give rise to limiting measure, see {\it e.g.} \cite{GK, GK2, D2}. Finally, the DOS functional computed for certain non unitary unitary band matrices, whose spectral properties are studied in \cite{HJ2}, and whose finite volume restrictions display similar features as those of the NSA model.

\section{Functional Calculus for CNU Contractions}\label{fccnuc}

We recall from \cite{SNF} the main properties of the functional calculus developed for contractions. 

Let $T$  be a contraction on a separable Hilbert space $\cH$. Consider the unique decomposition 
\be\label{ucnu}
T=T_0\oplus T_1 \ \ \mbox{on}\ \ \cH=\cH_0\oplus \cH_1,
\ee 
where $\cH_0=\{\psi\in \cH \ | \ \|T^n\psi\|=\|\psi\|=\|T^{*n}\psi\|, n\in \N\}$, $\cH_1=\cH\ominus\cH_0$ and $T_j=T|_{\cH_j}$, $j=0,1$, 
such that $T_0$ is unitary, and $T_1$ is completely non unitary (cnu).
The analysis of random unitary operators is by now well known, so we restrict attention to cnu contractions and therefore assume that $T=T_1$ in the following. We recall here some basic facts from harmonic analysis. In the following, $\D$ denotes the open disk, its boundary is $\partial\D$ that we will also identify with the torus $\T$. The set of holomorphic functions on an open set $S\subset\C$ is denoted by $\hol(S)$.

The Hardy class $H^p(\D)$ consists in holomorphic functions on $\D$ such that 
\be
\|u\|_p=\left\{
\begin{array}{ll}
\sup_{0<r<1}\left[\frac{1}{2\pi}\int_\T |u(re^{it})|^p dt\right] & \mbox{if} \ \ 0<p<\infty, \\
\sup_{z\in \D}|u(z)| &  \mbox{if} \ \ p=\infty.
\end{array}
\right.
\ee
For all $0<p\leq \infty$, functions in $H^p(\D)$ admit radial limits $\lim_{r\ra 1^-}u(re^{it})=u(e^{it})$ on $\partial\D$ for almost every $t\in \T$ and 
$\ln|u(e^{it})|\in L^1(\T)$. For $0<p<\infty$, $u\in H^p(\D)$ further satisfies $u(re^{i\cdot})\ra u(e^{i\cdot})$ in $L^p(\T)$ norm. Moreover, introducing $L^p_+(\T)$ as the set of functions in $f\in L^p(\T)$ whose negative Fourier coefficients all vanish, the spaces $H^p(\D)$ and $L^p_+(\T)$ can be isometrically identified, for all $1\leq p\leq \infty$. These function spaces are Banach spaces, and even Hilbert spaces whenever $p=2$. For later reference, we also introduce  for all $0<p\leq \infty$,  $H_0^p(\D)=\{g\in H^p(\D) \ | \ g(0)=0\}$. Finally, the disk algebra $A(\D)$ is defined as the set of continuous functions on $\bar \D$ that are holomorphic on $\D$, {\it i.e.} $A(\D)=H^\infty(\D)\cap C(\T)$. Let $\tilde{\cdot}$ denote the involution on holomorphic functions on $\D$ given by
$
\tilde{f}(z)= \overline{f(\bar{z})}$,
and set for any $0<r<1$
$
u_r(e^{it})=u(re^{it}).
$
The following statements, among other things, are proven in \cite{SNF}, Section III.2, Theorem 2.1. 
\begin{thm}\label{fcc}
Assume $T$ is a cnu contraction on a separable Hilbert space $\cH$. Then, for any $u\in H^\infty(\D)$, s.t. $u(z)=\sum_{n\geq 0}c_n z^n$ for $z\in\D$, $u(T)$ is defined by the strong limit
\be
u(T)=\mbox{s-}\lim_{r\ra 1^-}u_r(T), \ \ \mbox{where}\ \ u_r(T)=\sum_{n\geq 0}c_n r^nT^n,
\ee
which exists. The map $H^\infty(\D)\ni u \mapsto u(T)\in \cB(\cH)$ is an algebra homomorphism which further satisfies :
\bea
a) &&u(T)=\left\{\begin{array}{ll}
\I & \mbox{if} \ u(z)\equiv 1   \\
T & \mbox{if} \ u(z)=z 
\end{array}\right.  \nonumber \\
b) &&\|u(T)\|\leq \|u\|_\infty  \nonumber \\
c) &&u_n(T)\ra u(T)\ \ \mbox{in norm, resp. strong, resp. weak sense if }  \nonumber \\
&& u_n\ra u \ \mbox{uniformly on $\D$, resp. boundedly  a.e. on $\T$, resp. boundedly on $\D$}  \nonumber \\
d) && u(T)^*=\tilde{u}(T^*).
\eea
\end{thm}
\begin{rems}\label{fr}
i) By the identification of $H^\infty(\D)$ and $L^\infty_+(\T)$, this functional calculus can be viewed as a homomorphism $L^\infty_+(\T)\ni f\ra f(T)\in \cB(\cH)$, with 
\bea\label{cflp}
&&f(T)=\left\{\begin{array}{ll}
\I & \mbox{if} \ f(t)=1 \ \mbox{a.e.}   \\
T & \mbox{if} \ f(t)=e^{it} \ \mbox{a.e.} 
\end{array}\right.  \nonumber \\ 
&& \|f(T)\|\leq \|f\|_\infty .
\eea
ii) Conditions a), c) in the strong sense and (\ref{cflp}) make this functional calculus maximal and unique, see \cite{SNF}.\\
iii) If $u\in A(\D)$, $\|u_n(T)-u(T)\|\ra 0$, as $n\ra\infty$.\\
iv) If $T=T_0\oplus T_1$, as in (\ref{ucnu}), with unitary part $T_0$ having purely absolutely continuous spectrum, then Theorem \ref{fcc} holds with c) in the strong sense. See \cite{SNF}, Theorem 2.3.
\end{rems}

\section{DOS Functional}\label{genass}

We deal here with random contractions $T_\omega$ defined on a separable Hilbert space $\cH$ with ergodic properties we express as follows. 

We first assume some regularity assumptions.
Let  the probability space $(\Omega, {\cal F}, \P)$, where $\Omega$ is identified with $\{{X}^{\Z^d} \}$, 
$X\subset \R$, $d\in \N$, and $\P=\otimes_{k\in\Z^d}d\mu$, where $d\mu$ is a probability distribution 
on $X$ and ${\cal F}$ is the $\sigma$-algebra generated by the cylinders. We assume that 
$\Omega\ni\omega\mapsto T_\omega\in \cB(\cH)$ is measurable i.e.
\be\label{meast}
\forall \ffi, \psi\in \cH, \ \Omega\ni\omega\mapsto\bra \ffi | T(\omega)\psi\ket \ \  \mbox{is measurable}.
\ee 

Ergodicity is expressed in the following framework. We consider $\cH=l^2(\Z^d)$ and consider 
for $j\in \Z^d$ the shift operator $S_{j}$ defined on $\Omega$ by
\be
  S_j(\omega)_k=\omega_{k+j},  \ \  k\in\Z^d, \ \ \mbox{where} \ \omega_k\in X,
\ee
so that the measure $\P$ is ergodic under the set of commuting translations $\{S_j\}_{j\in\Z^d}$.
Let $\{\ffi_j\}_{j\in\Z^d}$ denote the canonical basis of $\cH$ and  $V_j$ be the unitary operator defined by
\be
V_j\ffi_k=\ffi_{k-j}, \ \ \forall k\in\Z^d.
\ee
We further assume the existence of a periodic lattice $\Gamma\subset \Z^d$ spanned by $\{\gamma_i\}_{i\in\{1,2,\dots,d\}}$, $\gamma_i\in\Z^d$ and the corresponding primitive cell $B=\{\sum_{i=1}^d x_i\gamma_i, 0\leq x_i<1\}\cap \Z^d$ so that for any $j\in\Z^d$, there exist a unique $b\in B$ and a unique $g \in\Gamma$ with $j=b+g$.

The random contractions we consider are ergodic in the following sense:
\be\label{ero}
T_{S_g (\omega)}=V_g T_\omega V_g^{-1}, \ \ \forall \ g\in\Gamma.
\ee

\begin{definition}\label{defdos} The DOS functional $L: H^{\infty}(\D)\ra \C$ is defined for all $f\in H^{\infty}(\D)$ by
\be
L(f)=\frac{1}{|B|}\sum_{b\in B}\E(\bra \ffi_b| f(T_\omega) \ffi_b\ket), 
\ee
where $ |B|$ denotes the cardinal  of $B$ and $T_\omega$ is a measurable cnu random contraction.
\end{definition}
We first note that
\begin{prop}\label{repfun}
The map $L: H^{\infty}(\D)\ra \C$ is a bounded and admits the following integral representation: there exists $\ffi\in L^1(\T)\backslash H_0^1(\D)$, such that for all $f\in H^{\infty}(\D)$
\be\label{repffi}
L(f)=\int_{\T}f(e^{it})\ffi(e^{it})\frac{dt}{2\pi}.
\ee
\end{prop}
\begin{rems}
0) The set $H_0^1(\D)$ should be understood as the set of boundary values of these functions. \\ 
i) The ergodicity assumption (\ref{ero}) plays no role here.\\
ii) Such functionals on $H^{\infty}(\D)$ are called {\em weakly continuous}, see {\it e.g. } \cite{G}, Section V.\\
iii)  The representation (\ref{repffi}) says that for any $G\in H_0^1(\D)$
\be
\int_{\T}f(e^{it})(\ffi+G)(e^{it})\frac{dt}{2\pi}=\int_{\T}f(e^{it})\ffi(e^{it})\frac{dt}{2\pi}, \ \ \forall f\in H^\infty(\D),
\ee
see \cite{G}, Theorem 5.2. It shows that the relevant 
information is carried by the negative Fourier coefficients only. \\
iv) In case  $\ffi\in L^p(\T)\subset L^1(\T)$, with $p>1$, we can actually represent $L$ by a unique $\ffi^-\in L^p$ obtained from $\ffi$ by substracting the contribution from the sum over positive Fourier coefficients. We give conditions for this to hold in Theorem \ref{smoothfi} below.\\
v) Making the dependence on $T$ of the functional $L$ explicit in the notation, we have for all $f\in H^\infty(\D)$,
\be
L_T(f)=\overline{L_{T^*}(\tilde f)}.
\ee
\end{rems}
\begin{proof} 
Linearity and the bound $|L(f)|\leq \|f\|_\infty$ stem directly from the properties of the functional calculus recalled above. Then one makes use of the 
following equivalence, see \cite{G}, Theorem 5.3: $L$ is a weakly continuous functionals on $H^\infty(\D)$ iff $L$ is continuous under bounded pointwise convergence; {\it i.e.} if $f_n\in H^\infty(\D)$, $\|f_n\|_\infty \leq M$ and $f_n(z)\ra f(z)$, for all $z\in\D$, then $L(f_n)\ra L(f)$. For such a sequence $f_n$, we have by point c) of Thm. \ref{fcc} that $\bra \ffi_k|f_n(T_\omega) \ffi_k\ket\ra \bra \ffi_k|f(T_\omega) \ffi_k\ket$ for all $k\in \Z^d$, and $|\bra \ffi_k|f_n(T_\omega) \ffi_k\ket|\leq M$, uniformly in $n, k, \omega$. Hence by Lebesgue dominated convergence, we also have $\E(\bra \ffi_b|f_n(T_\omega) \ffi_b\ket)\ra \E(\bra \ffi_b|f(T_\omega) \ffi_b\ket)$, for any $b\in B$,  which yields $L(f_n)\ra L(f)$, for $|B|$ finite. \ep
\end{proof}
\begin{rems}
i) The bound $|L(f)|\leq \|f\|_\infty$ is saturated: 
$
L(1)=1.
$
\\
ii) For any $j\in\N$, the function $z\mapsto z^j \in H^\infty(\D)$ and, denoting the Fourier coefficients of $\ffi$ by $\{\hat \ffi(k)\}_{k\in\Z}$,
\be
L({\cdot }^j)=\int_\T e^{ijt}\ffi(e^{it})\frac{1}{2\pi}=\hat\ffi(-j), \ \ \mbox{with } \ \ \hat\ffi(0)=1.
\ee 
\end{rems}

Then we observe that for functions in $A(\D)\subset H^\infty(\D)$,  we get an alternative representation of $L(\cdot)$ on the disk.\bigskip
 
Let us denote by $P[\ffi](re^{it})=P[\ffi](x,y)$ the harmonic function in $ \D$ given by the Poisson integral of $\ffi\in L^1(\T)$, with the usual abuse of notation. Due to the fact that $\lim_{r\ra 1^{-}}P[\ffi](re^{it})=\ffi(e^{it})$ almost everywhere and in $L^1(\T)$ norm, we can approximate $L(f)$ by an integral over smooth functions:  for any $f\in A(\D)$, and any $0<r<1$,
\bea\label{harm1}
L(f)&=&\int_{\T}f(e^{it})P[\ffi](re^{it})\frac{dt}{2\pi}+l_1(r), \ \mbox{where} \nonumber\\
 |l_1(r)|&\leq& \|f\|_\infty \|\ffi(\cdot)-P[\ffi](r\cdot)\|_{L^1(\T)}\ra 0 \ \mbox{as} \ r\ra 1^-.
\eea
Since $A(\D)$ consists in uniformly continuous functions, we can further approximate $f(e^{it})$ by $f(r{e^{it}})\equiv f_r(e^{it})$ to get
\bea\label{harm2}
L(f)&=&\int_{\T}f(re^{it})P[\ffi](re^{it})\frac{dt}{2\pi}+l_2(r), \ \mbox{where} \nonumber\\
 |l_2(r)|&\leq& \|f\|_\infty \|\ffi(\cdot)-P[\ffi](r\cdot)\|_{L^1(\T)}+\|f-f_r\|_{H^\infty}\|\ffi\|_{L^1(\T)}\ra 0 \ \mbox{as} \ r\ra 1^-.
\eea

The latter approximation allows us to provide $L(f)$ with a smooth integral representation over the disk.
\begin{prop}\label{harmden}
Let $\ffi\in L^1(\T)$ be the integral representation of $L(f)$ and $P[\ffi]$ its Poisson integral on $\D$. Then there exists a harmonic function $m_\ffi$ on $\D$ such that
for all $f\in A(\D)$,
\bea\label{intrepdis}
L(f)&=&\int_\D f(x+iy)m_\ffi(x,y)dxdy, \ \mbox{where} \nonumber \\
m_\ffi(x,y)&=&\frac{1}{2\pi}(\partial_x+i\partial_y)\{(x-iy)P[\ffi](x,y)\}.
\eea
\end{prop}
\begin{proof}
Consider the approximation (\ref{harm2}). By Stokes theorem applied to $f(z)P[\ffi](z)\in C^\infty(\D)$, we have
\bea
&&\int_{\T}f(re^{it})P[\ffi](re^{it})\frac{dt}{2\pi}=\int_{r\partial \D} f(z)P[\ffi](z)\bar z\frac{dz}{r^2 2i\pi}\nonumber\\
&&=\int_{r\D}\frac{\partial}{\partial \bar z}\left\{f(z)P[\ffi](z)\bar z\right\}\frac{d\bar z \wedge dz}{r^2 2i\pi}=\int_{r\D}f(z)\frac{\partial}{\partial \bar z}\left\{P[\ffi](z)\bar z\right\}\frac{dxdy}{r^2 \pi}.
\eea
Thanks to (\ref{harm2}), we can take the limit $r\ra 1^-$ which yields 
\be
m_\ffi(z)=\frac1\pi\frac{\partial}{\partial \bar z}\left\{P[\ffi](z)\bar z\right\}=\frac1\pi\left(P[\ffi](z)+\left(\frac{\partial}{\partial \bar z}P[\ffi](z)\right)\bar z\right).
\ee
 Using the fact that $P[\ffi]$ is harmonic, one finally gets
\be
\frac{\partial}{\partial  z}m_\ffi(z)=\frac1\pi\frac{\partial}{\partial  z}P[\ffi](z), \ \ \mbox{and} \ \ \frac{\partial^2}{\partial \bar z\partial  z}m_\ffi(z)=0.
\ee
\ep
\end{proof}
\begin{rem} In keeping with the fact that $\ffi\mapsto \ffi+G$, where $G\in H_0^1(\T)$ does not change the representation, one checks that
\be
m_{\ffi +G}(z)= m_\ffi(z) + G(z)/\pi, \ \mbox{where} \ \int_\D f(z)G(z)dxdy=0.
\ee
\end{rem}

The integral representations of $L(\cdot)$ discussed so far are intrinsic. There are of course many alternative integral representations in the disk: 
since $A(\D)\subset C(\bar\D)$, we can extend $L$ to $\hat L: C(\bar\D)\ra\C$, by means of Hahn-Banach Theorem, with $\|L\|=\|\hat L\|$.  Given $\hat L$, since $\bar\D$ is compact, the Riesz Representation Theorem asserts the existence of a unique complex Borel measure $d\mu$ on $\D$ such that $\hat L(f)=\int_{\bar \D} f(x,y) d\mu(x,y)$ for all $f\in C(\bar\D)$ and $\|\hat L\|=|\mu|(\D)$. Thus, 
\begin{lem} Let $T_\omega$ be a measurable random contraction. There exists a complex Borel measure $d\mu$ on $\bar\D$ such that $|\mu|(\bar \D)=1$ such that
\be
L(f)=\int_{\bar\D} f(x+iy) d\mu(x,y), \ \ \forall  f\in A(\D).
\ee
\end{lem}
\begin{rem}  The measure $d\mu$ uniquely determined by the extension $\hat L$ of $L$ to $C(\bar \D)$.
The example discussed in Section \ref{nubm} illustrates the fact that there may be infinitey many such integral representations of $L$.
\end{rem}

\section{Finite Volume Approximations}\label{sfva}

Let $\Lambda\subset \Z^d$ be given by $(2n+1)^d$ symmetric translates of $B$ along $\Gamma$ of the form 
\be\label{lan}
\Lambda=\bigcup_{n_i\in\Z \atop -n\leq n_i\leq n}B+\sum_{i=1}^d n_i\gamma_i\equiv \bigcup_{m=1,\cdots,(2n+1)^d} B+g_m, 
\ee
with $|\Lambda|=(2n+1)^d|B|$, and let
\be
\cH_\Lambda={\mbox{span }}\{\ffi_j,  \ j\in \Lambda  \}, \ \ \cH_{\Lambda^C}=\cH\ominus \cH_\Lambda,
\ee
together with the corresponding orthogonal projections onto these subspaces $P_\Lambda, P_{\Lambda^C}$.
\bigskip

Dropping $\omega$ from the notation for now, let assume that the cnu contraction $T$ can be written as
\be\label{deff}
T=T^\Lambda\oplus T^{\Lambda^C}+F^\Lambda, 
\ee
where $T^\Lambda$ and $T^{\Lambda^C}$ are defined on $ \cH_\Lambda$ and $ \cH_{\Lambda^C}$ and $F^\Lambda$ is a trace class operator on $\cH$, and furthermore $T^\Lambda$  is a cnu contraction. 
Such  decompositions can be obtained for example by setting
\bea
T^\Lambda&=&T|_{\cH_\Lambda}+ \ \mbox{boundary conditions}\\
T^{\Lambda^C}&=&T|_{\cH_{\Lambda^C}}+ \ \mbox{boundary conditions},
\eea
for suitable boundary conditions at $\partial \Lambda$, see below. For any $f\in H^\infty(\D)$, the operators $f(T), f(T^\Lambda)$ are well defined by functional calculus and we consider two random functionals on $H^\infty(\D)$ given by
\be
L_{\Lambda}(f)=\frac{1}{|\Lambda |}\tr (f(T^\Lambda)), \ \ \
\tilde{L}_\Lambda(f)=\frac{1}{|\Lambda |}\tr (P_\Lambda f(T)P_\Lambda).
\ee
From the bound $\|A\|_1\leq \mbox{rank} (A)\ \|A\|$ on the trace norm $\|\cdot \|_1$,  we deduce that for all $f\in H^\infty(\D)$
\be
\|L_{\Lambda}(f)\|\leq \|f\|_\infty, \ \ \|\tilde L_{\Lambda}(f)\|\leq \|f\|_\infty,
\ee
so that all arguments of the proof of Proposition \ref{repfun} apply. Hence $L_\Lambda$ and $\tilde L_\Lambda$ can be written in the form (\ref{repffi}) with corresponding random $L^1(\T)$ functions $\ffi_\Lambda(e^{i\cdot})$ and $\tilde \ffi_\Lambda(e^{i\cdot })$. More precisely we have
\begin{lem} Let $\lambda_j$, $j=1,...,|\Lambda|$, be the eigenvalues of the cnu contraction $T^\Lambda$, repeated according to their algebraic multiplicities.
Then, $\forall \ f\in H^\infty(\D)$, $L_{\Lambda}(f)=\frac{1}{2\pi}\int_{\T}f(e^{it})\ffi_\Lambda(e^{it})dt$ where
\bea\label{fila}
\ffi_\Lambda(e^{it})&=&\frac{1}{|\Lambda|}\sum_{j=1}^{|\Lambda|}\frac{1}{1-\lambda_j e^{-it}}.
\eea
\end{lem}
\begin{proof}
The finite dimensional contraction $T^\Lambda$ being cnu, $\sigma(T_\Lambda)\subset \D$. Hence, for any $f\in H^\infty(\D)$, and any $\rho<1$ large enough
\bea\label{repker}
L_{\Lambda}(f)&=&\frac{1}{|\Lambda|}\sum_{j=1}^{|\Lambda|}f(\lambda_j)=
\frac{1}{|\Lambda|}\sum_{j=1}^{|\Lambda|}\frac{1}{2i\pi}\int_{\rho\partial \D}\frac{f(z)}{z-\lambda_j}dz\nonumber \\
&=&\frac{1}{|\Lambda|}\sum_{j=1}^{|\Lambda|}\frac{\rho}{2\pi}\int_{\T}\frac{f_\rho(e^{it})}{\rho -\lambda_j e^{-it}}dt.
\eea
For each $j$, since $f\in H^\infty(\D)$, Lebesgue dominated convergence implies that the limit $\rho\ra1$ exists, which yields the result. \ep
\end{proof}
\begin{rems}
i)  An application of Stokes theorem shows that 
\be
L_{\Lambda}(f)=\frac{1}{|\Lambda|}\sum_{j=1}^{|\Lambda|}\frac1{\pi}\int_{\D} \frac{f(z)}{(1-\lambda_j \bar z)^2}dxdy,
\ee
an expression of the fact that $\frac1{\pi}(1-\lambda_j \bar z)^{-2}$ is the Bergman reproducing kernel.
 \\
ii) Introducing the normalised counting measure of $\sigma(T^\Lambda)$, $dm^\Lambda$  on the closed unit disc $\bar \D$ by
\be\label{countmeas}
dm^\Lambda(x,y)=\frac{1}{|\Lambda|}\sum_{j=1}^{|\Lambda|}\delta(x-\Re\lambda_j)\delta(y-\Im \lambda_j), \ \ \lambda_j\in \sigma(T^\Lambda), 
\ee
where the eigenvalues are repeated according to their algebraic multiplicities, we can extend $L_\Lambda$ to $C(\bar \D)$ by
\be
\hat L_{\Lambda}(f)=\int_{\bar \D}f(x,y)dm^\Lambda(x,y), \ \ \forall f\in C(\bar \D).
\ee
However, $\tilde{L}_\Lambda(f)$ does not make sense for $ f\in C(\bar \D)$.

\end{rems}

\subsection{Infinite Volume Limit}\label{ivl}

Restoring the variable $\omega$ in the notation for a moment, we show that the random functionals $L_{\Lambda, \omega}(f)$ and $\tilde L_{\Lambda, \omega}(f)$ converge almost surely to the deterministic DOS functional $L(f)$ as $\Lambda \ra \Z^d$, when $f\in A(\D)\subset H^\infty(\D)$. By $\Lambda \ra \Z^d$ or $|\Lambda|\ra\infty$, we mean $n\ra\infty$ in definition (\ref{lan}). This is done along the same lines as in the unitary case under ergodicity assumption, see \cite{J1}, for example.\bigskip

We start by a deterministic statement:
\begin{prop}\label{llt} Assume $T$ and  $T^\Lambda$ given by (\ref{deff}) are cnu contractions.\\
 If $\|P_\Lambda(T-T^\Lambda \oplus T^{\Lambda^C})\|_1=o(|\Lambda|)$ as $|\Lambda|\ra\infty$, then , for all $f\in A(\D)$,
\be
\lim_{|\Lambda|\ra\infty}L_{\Lambda}(f)-\tilde L_{\Lambda}(f)= 0.
\ee
\end{prop}
\begin{proof} We need to show that
\be
\lim_{|\Lambda|\ra\infty}\frac{1}{|\Lambda|}\left\{\tr (f(T^\Lambda) - \tr (P_\Lambda f(T) P_\Lambda)\right\}=0.
\ee
As $f(T^\Lambda) = P_\Lambda f(T^\Lambda) P_\Lambda$, cyclicity of the trace yields
\be
\tr (f(T^\Lambda) - P_\Lambda f(T) P_\Lambda)=\tr((f(T^\Lambda)- f(T))P_\Lambda)= \tr(P_\Lambda (f(T^\Lambda)- f(T))).
\ee
On the other hand, since $A(\D)$ consists in uniformly continuous functions on $\bar \D$, we have
\be
f(z)=\sum_{j=0}^N a_j z^j + R_N(z), \ \ \mbox{where} \ \  \sup_{z\in  \bar \D}|R_N(z)|=\|R_N\|_\infty\ra 0 \ \mbox{as} \  N\ra\infty.
\ee
Hence, together with the bound $\|A\|_1\leq \mbox{rank} (A)\ \|A\|$, we get the uniform estimate
\be\label{unibl}
\frac{1}{|\Lambda|}|\tr (R_N(T^\Lambda) - P_\Lambda R_N(T) P_\Lambda)|\leq \frac{\mbox{rank} (P_\Lambda)}{|\Lambda|} \|R_N(T^\Lambda) -R_N(T) \|\leq 2\|R_N\|_\infty.
\ee
Therefore, we can focus on $f(z)=z^j$, $j\in \N$. With 
\bea
(T^j-(T^\Lambda\oplus T^{\Lambda^C})^j)P_\Lambda&=&\sum_{k=0}^{j-1} T^k(T-T^\Lambda\oplus T^{\Lambda^C})(T^\Lambda\oplus T^{\Lambda^C})^{j-k-1}P_\Lambda\\
&=&\sum_{k=0}^{j-1} T^k(T-T^\Lambda\oplus T^{\Lambda^C})P_\Lambda(T^\Lambda\oplus \mathbb O)^{j-k-1}
\eea
we get for all $j\leq N$, using  $T$ and $T^\Lambda$ are contractions and cyclicity of the trace,
\be
\frac{1}{|\Lambda|}|\tr(P_\Lambda(T^j- (T^\Lambda\oplus T^{\Lambda^C})^j))|\leq \frac{j}{|\Lambda|}\|P_\Lambda(T-T^\Lambda \oplus T^{\Lambda^C})\|_1.
\ee
This estimate, the assumption $\|P_\Lambda(T-T^\Lambda \oplus T^{\Lambda^C})\|_1=o(|\Lambda|)$ and (\ref{unibl}) end the proof. \ep
\end{proof}

We finally turn to the infinite volume limit of $\tilde L_{\Lambda, \omega}(f)$, for $f\in A(\D)$, under the ergodicity assumption (\ref{ero}) on the way
randomness enters the contraction $T_\omega$.

\begin{prop}\label{limth}
Assume $T_\omega$ is ergodic in the sense of (\ref{ero}). For all $f\in A(\D)$, 
\be
\lim_{|\Lambda|\ra\infty}\tilde L_{\Lambda, \omega}(f)=L(f), \ \ \mbox{almost surely.}
\ee
\end{prop}
\begin{proof}
By construction of $\Lambda$, see (\ref{lan}), for any $H:\Z^d\ra \C$,  
\be
\sum_{k\in\Lambda}H(k)=\sum_{b\in B}\sum_{m=1}^{(2n+1)^d} H(b+g_m), \ \ \mbox{where} \ \ 
g_m=\sum_{i=1}^d n_i\gamma_i\in\Gamma.
\ee 
Thus, by the ergodicity assumption,
\bea
\tilde L_{\Lambda, \omega}(f)&=&\frac{1}{|\Lambda|}\sum_{k\in\Lambda}\bra \ffi_k| f(T_\omega) \ffi_k\ket
=\frac{1}{|\Lambda|}\sum_{b\in B}\sum_{m=1}^{(2n+1)^d}\bra \ffi_{b+g_m}| V_{g_m}^*f(T_{S_{g_m}(\omega)}) V_{g_m}\ffi_{b+g_m}\ket\nonumber \\
&=&\frac{1}{|\Lambda|}\sum_{b\in B}\sum_{m=1}^{(2n+1)^d}\bra\ffi_b| f(T_{S_{g_m}(\omega)}) \ffi_b\ket=
\frac{1}{|\Lambda|}\sum_{b\in B}\sum_{|n_i|\leq n}\bra\ffi_b| f(T_{S^{n_1}_{\gamma_1}\cdots S^{n_d}_{\gamma_d}(\omega)}) \ffi_b\ket.
\eea
Thanks to Birkhoff Theorem, we have on a set $\Omega_f\subset \Omega$ of measure one, and for all 
$b\in B$,
\be
\lim_{n\ra\infty} \frac{1}{(1+2n)^d}\sum_{|n_i|\leq n}\bra\ffi_b| f(T_{S^{n_1}_{\gamma_1}\cdots S^{n_d}_{\gamma_d}(\omega)}) \ffi_b\ket=\E(\bra\ffi_b| f(T_{\omega}) \ffi_b\ket).
\ee
Since $A(\D)$ is separable, the statement above is true for a dense countable set of functions $\{f_m\}_{m\in \N}$ on $\cap_{m\in \N} \Omega_{f_m}=\Omega_0\subset \Omega$, a set of measure one. Since $|B|<\infty$, we infer 
\be
\lim_{|\Lambda|\ra\infty}\frac{1}{|\Lambda|}\sum_{k\in\Lambda}\bra \ffi_k| f(T_\omega) \ffi_k\ket=\sum_{b\in B}\frac{1}{|B|}\E(\bra\ffi_b| f(T_{\omega}) \ffi_b\ket) \ \ \mbox{a.s.}
\ee
which proves the statement. \ep
\end{proof}
\begin{rems}\label{remlimth}
i) We consider $A(\D)$ only, a separable space, since $H^\infty(\D)$ is not.\\
ii) As $\left|\frac{1}{|\Lambda|}\sum_{k=1}^{|\Lambda|}\bra \ffi_k| f(T_\omega) \ffi_k\ket\right| \leq \|f\|_\infty$, Lebesgue dominated convergence implies that
\be\label{expval}
\lim_{|\Lambda|\ra\infty}\E(\tilde L_{\Lambda, \omega}(f))=\lim_{|\Lambda|\ra\infty}\frac{1}{|\Lambda|}\sum_{k=1}^{|\Lambda|}\E(\bra \ffi_k| f(T_\omega) \ffi_k\ket)=L(f).
\ee
iii) Under the assumptions of Propositions \ref{llt} and \ref{limth}
\be\label{thelim}
\lim_{|\Lambda|\ra\infty} L_{\Lambda, \omega}(f)=L(f), \ \ \mbox{almost surely.}\\
\ee
iv) The result also holds if $T=T_0\oplus T_1$ with purely absolutely continuous unitary part $T_0$.
\end{rems}

\subsection{ $\spr(T^\Lambda)<1$}
We show that if an {\it a priori} uniform estimate on the spectral radius of $T_\Lambda$ holds, we deduce anti-analyticity of the integral representation $\ffi$ of $L(f)$.\bigskip

Let us drop the dependence on $\omega\in \Omega$ in  the notation. The form $L_\Lambda$ is represented by integration against $\ffi_\Lambda(e^{it})$, see (\ref{fila}), which can be written with $z=e^{it}$ as
\be
\ffi_\Lambda(z)=\overline{\psi_\Lambda(z)}, \ \ \mbox{where} \ \psi_\Lambda(z)=\frac{1}{|\Lambda|}\sum_{j=1}^{|\Lambda|}\frac{1}{1-\bar \lambda_j z}. 
\ee 
As $T_\Lambda$ is cnu, $\psi_\Lambda$ is holomorphic in $\D$ and we have the absolutely converging power series
\be
\psi_\Lambda(z)=\sum_{k=0}^\infty \frac{\tr\ {T^*_\Lambda}^k}{|\Lambda|} z^k, \ \ \forall z\in \D.
\ee
\begin{thm} \label{smoothfi}
Let $T_\omega$ satisfy the hypotheses of Propositions \ref{llt} and \ref{limth}. Assume there exists $r<1$ such that
\be\label{bsprt}
\spr(T_\Lambda)\leq r, \ \ \forall\  \Lambda \in \Z^d\ \mbox{and} \ \forall \ \omega\in \Omega.
\ee 
Then, 
 there exists $\psi\in { \hol} (\D/r)$ such that  Proposition \ref{repfun} holds with $\ffi(e^{it})=\overline{\psi(e^{it})}$:
\be\label{inreho}
L(f)=\int_{\T}f(e^{it})\overline{\psi(e^{it})}\frac{dt}{2\pi}.
\ee
\end{thm}
\begin{rems} i) We can  satisfy the hypothesis by properly rescaling the operator $T_\omega$.\\
ii) The property $\spr(T_\Lambda)\leq r<1$ $\forall \Lambda$ does not imply $\spr(T)<1$. Indeed, finite volume approximations of non normal operators typically miss important parts of $\sigma(T)$, see Section \ref{nsam} and (\ref{diffspr}) in particular.
\end{rems}
\begin{proof}
If (\ref{bsprt}) holds, $\psi_\Lambda$ is holomorphic in the larger disc $\D/r$ for all $\Lambda\in \Z^d$, $\omega\in\Omega$ and 
\be
|\psi_\Lambda(z)|\leq \frac{1}{1-r|z|}, \ \ \forall z\in \D/r.
\ee
In particular, the family $\{\psi_\Lambda\}_{\Lambda}$ of holomorphic functions on $\D/r$ is uniformly bounded  on each compact subset of $\D/r$. Hence, by Montel Theorem, see e.g. \cite{R}, Theorem 14.6, $\{\psi_\Lambda\}_{\Lambda}$ is a normal family. Therefore, for each fixed $\omega\in \Omega_0$, the set of measure one 
on which Proposition \ref{limth} holds, there exists a subsequence $\{\psi_{\Lambda_k}\}_{k\in\N}$ 
which converges uniformly on each compact subsets of $\D/r$ to a function $\psi(z)$ which is holomorphic on $\D/r$. In particular, for all $f\in A(\D)$,
\be
\lim_{k\ra\infty }L_{\Lambda_k}(f)=\lim_{k\ra\infty }\int_{\T}f(e^{it})\overline{\psi_{\Lambda_k}(e^{it})}\frac{dt}{2\pi}= \int_{\T}f(e^{it})\overline{\psi(e^{it})}\frac{dt}{2\pi},
\ee
where $\psi(z)$ is analytic in a neighbourhood of $\D$. By Remark \ref{remlimth} iii), we get that $\ffi\in L^1(\T)$ which represents $L(f)$ is given by
$
\ffi(t)=\overline{\psi(e^{it})},
$
and $\psi$ is independent of $\omega$.
\ep
\end{proof}

Consequently, 
\begin{cor} \label{scalprod}
With  $\psi=\sum_{n=0}b_nz^n \in \hol (\D/r)$, and for all $f(z)=\sum_{n=0}a_nz^n\in A(\D)$, 
\be\label{explpsi}
L(f)=\bra \psi | f\ket_{L^2_+(\T)}= \bra \hat {\psi} | \hat f\ket_{l^2(\N)}=\sum_{n=0}^\infty  \bar b_n a_n.
\ee
 The integral representation (\ref{intrepdis}) reads
\be
m_\ffi(z)=m_{\overline{\psi}}(z)=\frac1\pi\frac{\partial}{\partial \bar z}\left\{\overline{\psi(z)z}\right\}.
\ee

\end{cor}

\subsection{Representations of $L$ via finite volume approximations}

Let us consider now the random counting measure $dm_\omega^\Lambda$ (\ref{countmeas}) and assume that it admits a weak limit, almost surely:  for all $\omega\in\Omega_0$ with  $\P(\Omega_0)=1$, and for all $f\in C(\bar \D)$,
\be
\lim_{|\Lambda|\ra \infty}\hat L_{\Lambda,\omega}(f)=\lim_{|\Lambda|\ra \infty}\int_{\bar \D}f(x,y)dm_\omega^\Lambda(x,y)=\int_{\bar \D}f(x,y)dm_\omega(x,y).
\ee
Then, for any $\omega\in\Omega_0$, $dm_\omega\geq 0$ provides another representation of $L$ on $A(\D)$, since, specialising to $f\in A(\D)$, we get  from (\ref{thelim})
\be\label{replimcount}
L(f)=\int_{\bar \D}f(x+iy)dm_\omega(x,y).
\ee
\begin{thm}\label{convcount} Let $T_\omega$ satisfy the assumptions of Propositions \ref{llt} and \ref{limth} and $dm_\omega^\Lambda$ in (\ref{countmeas}) converge weakly to $dm_\omega$, almost surely. Then $L$ admits the following representation 
\be
L(f)=\int_{\bar\D} f(x+iy) d\bar m(x,y), \ \ \forall  f\in A(\D),
\ee
where $d\bar m \geq 0$ is given by $\E(dm_\omega)$.\\
Further assume $\spr (T_\omega^{\Lambda})\leq r<1$, for all $\omega\in \Omega$ and $\Lambda$. Then, Theorem \ref{smoothfi} holds with 
\be\label{borel}
\psi(z)=\int_{\bar\D}\frac{d\bar m(x,y)}{1-z(x-iy)}, \ \ \forall z \in \D/r,
\ee
and Proposition \ref{harmden} holds with
\be
m_{\bar\psi}(z)=\frac1\pi\int_{\bar\D} \frac{d\bar m(x,y)}{(1-\bar z (x+iy))^2}.
\ee
\end{thm}
\begin{rem}
 If the weak limit of the normalised counting measure of the finite volume spectrum exists, see {\it e.g.} \cite{GK2} for such cases, it provides another representation of the DOS functional. In that sense, the spectrum of the finite volume restrictions acquire a global meaning, in spite of the fact that it can be very different from the spectrum of the infinite volume  operator. In particular, the support of the limiting measure can be disjoint from the spectrum of the operator, see
Section \ref{nubm} for such an example. 
\end{rem}   
\begin{proof}
The first statement is a consequence of (\ref{replimcount}) and Fubini's Theorem. The assumption on $\spr(T_\omega^{\Lambda})$ implies supp~$d\bar m\subset \{|z|\leq r\}$.  Then, Corollary \ref{scalprod} applied to $f(z)=z^k$ yields the coefficients of the power expansion of the holomorphic function 
\be
\psi(z)=\sum_{k\geq 0}b_kz^k, \ \ z\in\D/r.
\ee
We have
\be
b_k=\int_{\bar\D}(x-iy)^kd\bar m(x,y), \ \ \mbox{with} \ |b_k|\leq r^k, \ \ k\in\N.
\ee
Thus, exchanging integration and summation, we get expression (\ref{borel}) for $z\in \D/r$. The last statement follows from Corollary \ref{scalprod}.
\ep
\end{proof}
\begin{rem}\label{extunit}
Thanks to Remarks \ref{fr} iv), and \ref{remlimth} iv), all results of Section \ref{sfva} hold if $T$ writes as $T=T_0\oplus T_1$, see (\ref{ucnu}), with a unitary part $T_0$ that is purely absolutely continuous. 
\end{rem}

\section{Special Cases}\label{illustr}

We take a closer look at various particular cases allowing us to get further informations on the integral representation $\ffi$.

\subsection{The Normal Case}\label{normal}

A first special case of interest occurs when $T_\omega$ is normal, {\it i.e.}, when there exists orthogonal projection valued measures $\{dE_\omega(x+iy)\}_{(x+iy)\in \sigma(T_\omega)}$ such that in the weak sense,
\be
T_\omega=\int_{\sigma(T_{\omega})} (x+iy) dE_\omega(x+iy).
\ee
In such a case, we have a continuous functional calculus: for any $f\in C(\bar \D)$
\be
f(T_\omega)=\int_{\sigma(T_{\omega})} f(x+iy) dE_\omega(x+iy).
\ee
Hence, assuming that $T_\omega$ is normal for any $\omega \in \Omega$,  Definition \ref{defdos} gives rise to a positive functional on $C(\bar \D)$, so that by Riesz-Markov Theorem  
\be
L(f)=\int_{\bar \D}f(x+iy)dm(x,y), \ \mbox{where $dm$ is a non-negative Borel measure on $\bar \D$}. 
\ee
In this favourable framework, we have 
\begin{lem}
Let $T_\omega, T_\omega^\Lambda$ be contractions such that $\|P_\Lambda(T-T_\omega^\Lambda \oplus T_\omega^{\Lambda^C})\|_1=o(|\Lambda|)$, uniformly in $\omega\in\Omega$. Further assume $T$ and $ T_\omega^\Lambda$ are normal and the ergodicity assumption (\ref{ero}) holds. Then, as $\Lambda\ra \infty$,
\be
dm^{\Lambda}_{\omega}\ra dm \ \ \mbox{a.s., in the weak-$*$ sense. }
\ee
\end{lem}
\begin{proof}
The same arguments using Stone Weierstrass and Birkhoff theorems together with the separability of $C(\bar \D)$ prove the result as in the previous section. \ep
\end{proof} \\

Applying Propositions \ref{repfun} and \ref{harmden} to the normal case, we get for any $f\in A(\D)$
\be
\int_{\bar \D}f(x+iy)m_\ffi(x,y)dxdy=\int_{\bar \D}f(x+iy)dm(x,y),
\ee
where $dm(x,y)$ is the non negative usual density of states, and $m_\ffi(x,y)$ is harmonic and in general complex valued. This special case makes explicit the lack of uniqueness in the representation of analytic functionals. 
\begin{rem}
If, moreover, there exists $0<r<1$ such that for all $\omega\in \Omega$, $\spr(T_\omega^\Lambda)\leq r$ and $T_\omega$ is cnu, then Theorem \ref{convcount} holds with $d\bar m=dm$, the density of states.
In particular,
\be
L(f)=\int_{\bar\D}\int_{\bar\D}\frac{f(z) d m(x,y)}{(1-\bar z(x+iy))^2}\frac{d\bar z\wedge dz}{2i\pi}, \ \ \forall z \in \D/r.
\ee
\end{rem}
\subsection{Multiple of Unitary Operators }
Consider now a special normal case where the statement above allows us to make the link between the DOS functional and the density of state measure of a random unitary operator $U_\omega$ more explicit. Let $0<r<1$ and $U_\omega$ be a random unitary operator defined on $l^2(\Z)$, with $\omega\in \T^{\Z}$,  which is measurable and ergodic. Models of this type are studied in \cite{bhj, J1}. The details do not matter for our purpose here. We consider the random normal cnu contraction
\be
T_\omega = r U_\omega.
\ee
We assume that there exist finite volume approximations such that $U^{\Lambda}_\omega$ is a finite dimensional unitary matrix; see e.g. \cite{J1} for examples of this situation  with $\Lambda=\{-n+1,n\}$. Consequently, we have the trivial bounds $\spr (T_\omega^{\Lambda})=\spr (T_\omega)=r$, $r<1$. Stressing the $r$ dependence in the notation, the DOS functional $L^{(r)}(\cdot)$ on $A(\D)$ is represented by 
\be\label{lru}
L^{(r)}(f)=\int_\T \overline{\psi^{(r)}(e^{it})}f(e^{it})\frac{dt}{2\pi}, \ \ \mbox{where $\psi^{(r)}(z)\in \hol (\D/r)$. }
\ee
The density of states measure $dk$ for ergodic unitary operators is a normalised positive regular Borel measure on $\T$ characterised as in Section \ref{normal} by Definition \ref{defdos} with $f$ continuous on the circle via Riesz-Markov Theorem, \cite{J1}:
\be\label{unitdos}
\frac{1}{2}\sum_{k=1}^2\E(\bra \ffi_k| f(U_\omega) \ffi_k\ket)=\int_\T f(e^{it})dk(t), \ \ \forall f\in C(\partial\D).
\ee
The explicit link is provided by
\begin{prop} Let $T_\omega=rU_\omega$, with $0<r<1$. With the notations and assumptions above,
for all $|z|<1/r$,
\be\label{repcau}
\psi^{(r)}(z)=\int_{\T} \frac{dk(t')}{1-zre^{-it'}}.
\ee
In particular,
\be\label{poisson}
\psi^{(r)}(e^{it})=\frac12+\frac12\int_{\T} \left[\frac{1+re^{i(t-t')}}{1-re^{i(t-t')}}\right]dk(t'),
\ee
and, 
\be\label{poissonk}
\Re \psi^{(r)}(e^{it})-\frac12=\frac12 P[dk](re^{it})\geq 0, 
\ee
where $P[dk](z)$ denotes the Poisson integral in $\D$ of the non negative measure $dk$ on $\T$.
\end{prop}
\begin{rems} 
i) The representation (\ref{repcau}) shows that $\psi^{(r)}(z)$ 
coincides with the Borel (or Cauchy) transform of the measure $dk$ on $\T$, taken at point $rz\in \D$.
The function $\psi^{(r)}$ may admit analytic extensions outside $\D/r$, depending on $dk$. \\
ii) While $T_\omega=rU_\omega$ is cnu, the boundary values as $r\ra 1^-$ of the real part of the integral representation of 
$L(f)$  yield the absolutely continuous component of the density of state measure of $U_\omega$. Indeed, if $dk(t)=g(e^{it})\frac{dt}{2\pi}+d\mu_s(t)$ is the Lebesgue decomposition of  $dk$, 
\be
\lim_{r\ra 1^-}2\Re \psi^{(r)}(e^{it})-1=g(e^{it}), \ \mbox{almost everywhere on $\T$},
\ee 
see \cite{R}, Theorem 11.24. 
\end{rems}
\begin{proof}
Let $f\in A(\D)$. From (\ref{lru}) and (\ref{unitdos}), we get
\be
\int_\T \overline{\psi^{(r)}(e^{it})}f(e^{it})\frac{dt}{2\pi}=\int_\T f(re^{it})dk(t).
\ee
The coefficients of the expansion
$\psi^{(r)}(z)=\sum_{n\in\N}b_nz^n, \ \ \mbox{for $z\in \D/r$}$
are given by $b_n= r^n \hat k (n),  n\in \N,$
where $\hat k (n)=\int_{\T} dk(t)e^{-itn}$.
Hence, with $\sign(0)=0$,
\be
\psi^{(r)}(e^{it})=\sum_{n\in\N} r^n\hat k(n) e^{int} \ \Rightarrow \ \left\{\begin{matrix}
2\Re \psi^{(r)}(e^{it})-1&=&\sum_{n\in\Z} r^{|n|}\hat k(n) e^{int}\\
2\Im  \psi^{(r)}(e^{it})&=&-i\sum_{n\in\Z} \sign(n)r^{|n|}\hat k(n) e^{int}
.\end{matrix}\right.
\ee
The last two series coincide with the real and imaginary parts of
$dk*(P_r(\cdot)+iQ_r(\cdot))(e^{it}),$ where  $P_r(t)$ and $Q_r(t)$, are the Poisson and conjugate Poisson kernels given by
\be
P_r(t)+iQ_r(t)=\frac{1+re^{it}}{1-re^{it}}=\frac{1-r^2+i2r\sin(t)}{1+r^2-2r\cos(t)}, \ \ 0<r<1.
\ee
Replacing $e^{it}$ in (\ref{poisson}) by $z\in \D/r$ yields (\ref{repcau}).
\ep
\end{proof}

\subsection{Non Self-Adjoint Anderson Model}\label{nsam}

The non self-adjoint Anderson model (NSA model for short), provides an example in which the finite volume and infinite volume versions of the random operator have quite different spectra that can be computed explicitly. See \cite{D2} and \cite{GK2} for more general non self-adjoint random operators with similar properties. After suitable rescaling, the NSA model provides us with an illustration of our results.\bigskip

The NSA model $H_\omega(g)$, is a one parameter deformation of the one dimensional random Anderson model of solid state physics defined as follows in the canonical basis of $l^2(\mathbb Z)$:
\begin{equation}
H_\omega(g) \varphi_j= e^{-g} \varphi_{j-1}+ e^g \varphi_{j+1}+\omega_j \varphi_j, \ \ \ \forall  j\in \mathbb Z.
\end{equation}
We assume that $g\geq 0$ and $\{\omega_j\}_{j\in\mathbb Z}$ are i.i.d. real valued random variables distributed according to a measure $d\mu$ supported on a compact interval $[-B,B]$. Let $E_g$ be the ellipse
\be
E_g=\{e^{g+i\theta} + e^{-(g+i\theta)}\ |  \ \theta \in [0, 2\pi]\},
\ee
which coincides with the spectrum of $H_0(g)$, the NSA model in absence of random potential.
It is proven in \cite{D2} that, provided
$B\geq e^{g} + e^{−g}$,
\be\label{snsa}
\sigma(H_\omega(g))=E_g + [-B, B], \ \ \mbox{almost surely.}
\ee
Moreover, consider $H^{\Lambda}_\omega(g)$ the finite volume restriction of $H_\omega(g)$ to $l^2(\Lambda)$, where $\Lambda=\{-n,n\}$ with Dirichlet boundary conditions.  Then, for $g>0$, the matrix $H^{\Lambda}_\omega(g)$ is similar to $H^{\Lambda}_\omega(0)$ with 
similarity transform given by
$W\ffi_k=e^{kg}\ffi_k$, so that its spectrum is real. The matrix $H^{\Lambda}_\omega(0)$ is the finite volume restriction 
of the Anderson model $H_\omega(0)$, with
\be\label{convev}
\sigma(H_\omega(0))=E_0 + [-B, B], \ \ \mbox{almost surely.}
\ee
To cast these considerations in our framework, we assume $B\geq e^{g} + e^{-g}$ and set
\be\label{defresnsa}
T_\omega(g)=H_\omega(g)/(e^{g} + e^{-g}+B), \ \mbox{resp.}\ \ T^{\Lambda}_\omega(g)=H^{\Lambda}_\omega(g)/(e^{g} + e^{-g}+B),
\ee
on $l^2(\Z)$, resp. $l^2(\Lambda)$ with Dirichlet boundary conditions. By construction, both operators are contractions.
Moreover
\begin{lem} \label{ldiffspr}
The matrix $T^{\Lambda}_\omega(g)$ is cnu for $|\Lambda|$ large enough, and all $\omega\in \Omega$, whereas 
the operator $T_\omega(g))$ is cnu almost surely. Also, for $g>0$,
 \be\label{diffspr} 
\spr (T_\omega(g))=1,  \ \ \mbox{a.s. and} 
\ \ \lim_{\Lambda\ra\infty}\spr (T^{\Lambda}_\omega(g))=\frac{2+B}{e^{g} + e^{-g}+B}<1, \ \ \forall \omega\in\Omega.
\ee
\end{lem}
\begin{proof}
Statements (\ref{diffspr}) are consequences of (\ref{snsa}), (\ref{convev}), and properties of finite volume approximations of self adjoint operators. Thus $T^{\Lambda}_\omega(g)$ is cnu for $|\Lambda|$ large enough.
Given (\ref{snsa}),  $T_\omega(g)$ is cnu if $\pm 1$ are not eigenvalues of
$T_\omega(g)$. Suppose $\ffi_\pm\in l^2(\Z)$ are normalized and satisfy $T_\omega(g)\ffi_\pm=\pm\ffi_\pm$. Then, 
using the fact that $T^*_\omega(g)$ is a contraction, and Cauchy-Schwarz inequality, we get $T^*_\omega(g)\ffi_\pm=\pm\ffi_\pm$.
Hence $\pm 1$ are eigenvalues with eigenvectors  $\ffi_\pm$ for the self-adjoint Anderson type operator $(T_\omega(g)+T^*_\omega(g))/2$. But this is known to happen with zero probability only. \ep
\end{proof}\\

Stressing the $g$ dependence in the notation, we denote  the DOS functional on $A(\D)$ by $L^{(g)}(\cdot)$. The foregoing and Theorem \ref{convcount} immediately show that
\begin{lem} Let $dm_\omega^{\Lambda}(g)$ be the normalised counting measure of $\sigma(T^{\Lambda}_\omega(g)))$. Then, for any $g>0$,
\be
dm_\omega^{\Lambda}(g)\ra dk_g, \ \ \mbox{a.s.,  in the weak sense,}
\ee
 where $dk_g\geq 0$ is the density of states of the rescaled self-adjoint Anderson model
 $\frac{H^{\Lambda}_\omega(0)}{e^{g} + e^{-g}+B}$. 
 
 Hence, for any $f\in A(\D)$, and any $g\geq 0$,
\be \label{lgdk}
L^{(g)}(f)
=\int_{[-2+B,2+B]}f\left(\frac{x+iy}{e^{g} + e^{-g}+B}\right)dk(x),
\ee
where $dk\geq 0$ is the density of states of  $H^{\Lambda}_\omega(0)$.
\end{lem}
\begin{rem}
The support of $dk_g$ is a subset of the almost sure spectrum of $T_\omega(g)$.  
\end{rem}
The estimates provided in Lemma \ref{ldiffspr} show that $L^{(g)}(\cdot)$ admits an integral representation on $\T$ in term of a holomorphic function $\psi^{(g)}(z)$, see Theorem \ref{smoothfi}. Moreover,
\begin{lem} Set  $s(g)=e^{g} + e^{-g}+B$.
The function $\psi^{(g)}$ is real analytic and admits a converging power series in $\{|z|<{s(g)}/{(2+B)}\}$.
Moreover, $\psi^{(g)}$ admits an analytic continuation on $\C\setminus\{[s(g)/(2+B), \infty[\cup ]-\infty, -s(g)/(2+B)]\}$ given by
\be\label{nicepsig}
\psi^{(g)}(z)=\int_{[-(2+B),2+B]}\frac{dk(x)}{1-zx/s(g)}.
\ee
\end{lem}
\begin{rem}
With the Borel transform $F_k$ of $dk$ given by
\be
F_k(z)=\int_{[-2+B,2+B]}\frac{dk(x)}{x-z}, \ \ \forall z\in \C\setminus [-2+B,2+B],
\ee
we have for $z\neq 0$,
\be
\psi^{(g)}(z)=-F_k\left({s(g)}/{z}\right){s(g)}/{z}.
\ee
\end{rem}
\begin{proof} It is a special case of Theorem \ref{convcount} where  the coefficients of the power expansion of $\psi^{(g)}$ are given by
\be
b_n=\overline b_n=\frac{\int_{[-2+B,2+B]} x^n dk(x)}{(e^{g} + e^{-g}+B)^n}, \ \ , |b_n|\leq \frac{(2+B)^n}{s(g)^n}\ \ \forall n\in\N.
\ee
\ep
\end{proof}
\subsection{Non Unitary Band Matrices}\label{nubm}
We further illustrate both the discrepancy between limiting measure and spectrum of the full operator, and the multiplicity of representations of the DOS functional by measures on $\bar \D$ by considering 
random contractions with a band matrix representation.\bigskip

Let $T_\omega$ be defined on $l^2(\Z)$ by its matrix representation in the canonical basis given as
\be\label{matrixt1}
T_\omega=\begin{pmatrix}
\ddots & e^{i\omega_{2j-1}}\gamma &e^{i\omega_{2j-1}} \delta & & & \cr
            &0                                          &0                                                   & & &\cr
             &0                                          &0      & e^{i\omega_{2j+1}}\gamma &e^{i\omega_{2j+1}}\delta &  \cr               
             & e^{i\omega_{2j+2}}\alpha &e^{i\omega_{2j+2}}\beta & 0& 0 &     \cr
             & & & 0                                          &0  &  \cr
           & &  & e^{i\omega_{2j+4}}\alpha &e^{i\omega_{2j+4}}\beta & \ddots
\end{pmatrix},
\ee
where $\{\omega_j\}_{j\in\Z}$ are $\T$-valued iid random variables. The deterministic coefficients $\alpha, \beta,\gamma, \delta$ are assumed to give rise  to a non unitary matrix  
\be\label{concon}
C_0=\begin{pmatrix}\alpha & \beta \cr \gamma & \delta \end{pmatrix} \ \mbox{s.t. } \ \ \|C_0\|_{\C^2}\leq 1.
\ee
Actually, $T_\omega= D_\omega T$, where $D_\omega$ is a diagonal unitary random operator with elements $e^{i\omega_j}$, and $T$ is a deterministic operator whose representation has the form (\ref{matrixt1}), where all $\omega_j$'s are equal to zero.
Such random ergodic operators arise in the analysis of certain random quantum walks and some of their spectral 
properties are analysed in \cite{HJ2}. In particular, it is shown there that $T_\omega$ is a cnu contraction on $l^2(\Z)$ if and only if 
\be \label{cnucon}
|\det C_0|<1, |\alpha|<1, \ \mbox{ and }\ |\delta|<1.
\ee
Moreover, $\|T_\omega\|=1$, and $\spr (T_\omega)$ may take the value 1, depending on the parameters.

We first note the following simple general result:
\begin{lem}\label{ldelt} Let $T_\omega$ defined by (\ref{matrixt1}) with (\ref{concon}), and (\ref{cnucon}), and assume the random phases $e^{i\omega_j}$ are uniformly distributed. Then
\be \label{delt}
L(f)=f(0),
\ \ \forall  f\in A(\D).
\ee 
\end{lem}
\begin{proof} 
One first observes that for any $n\in\N^*$, any $k\in\Z$, $\bra \ffi_k | T_\omega^n \ffi_k \ket=e^{i m(k) \omega_k}\tau(\hat \omega(k))$, where $\tau(\hat \omega(k))$ is a random variable independent of $\omega_k$, and $m(k)\in \N^*$. This is a consequence of the fact that no cancellation of the random phases can occur due to the shape  (\ref{matrixt1}) of the matrix representation of
$T_\omega$. Hence, $\E (\bra \ffi_k | T_\omega^n \ffi_k \ket)=0$, for all $k\in \Z$ and $n\in \N^*$. Thus, approximating any $f\in A(\D)$ by a polynomial, we get the result.\ep 
\end{proof}\\
\begin{rem}\label{xxx}
We get from (\ref{delt}) that the integral representations of $L(\cdot)$ on $\T$ and $\D$ are given by $\ffi\equiv 1$ and $m_\ffi\equiv \frac1\pi$ respectively.
\end{rem}

The DOS functional (\ref{delt}) can be extended to $C(\bar \D)$ as integration against a Dirac measure at the origin, but not only. For example, using in polar coordinates, defining for any $0<\rho<1$
\be
\hat L^{(\rho)}(f)=\frac{1}{2\pi}\int_{\bar\D} f(r\cos(\theta),r\sin(\theta))\delta_{\rho}(r)d\theta = \int_{\bar\D} f(x,y)d\mu^{(\rho)}(x,y),
\ee
we get a one parameter family of extensions of $L$ corresponding to integration against positive measures on $\D$, $d\mu^{(\rho)}$ with support of the circle of radius $\rho$, centered at the origin, such that $\hat L^{(\rho)}(f)=f(0)$ if $f\in A(\D)$. The support of $d\mu^{(\rho)}$ may or may not belong to the spectrum of $T_\omega$: the examples treated in \cite{HJ2} show that for $\det C_0=0$ and uniform i.i.d. phases $\omega_j$, the spectrum $T_\omega$ is given by the origin and ring centered at the origin whose radiuses depend on the parameters.  Note finally that any  radial probability measure with smooth density $\mu(r)$ provides us with an extension $\hat L(\cdot)$ on 
$C(\bar \D)$ of the form
\be
\hat L(f)=\int_{\bar\D} f(r\cos(\theta),r\sin(\theta))\mu(r)r dr d\theta,
\ee
which agrees with $L(\cdot)$ on $A(\D)$, thanks to the mean value property of harmonic functions.\\

We now specify the values of the parameters to $C_0=\begin{pmatrix} 1 & 0 \\ 0 & g\end{pmatrix}$, $0<g<1$, see (\ref{concon}), and look at finite volume restriction generated by finite rank perturbations. Note that in this case, $T_\omega$ is not cnu since
\be
T_\omega\simeq S\oplus g S,
\ee
where $\simeq$ denotes unitary equivalence and $S$ is the standard shift on $l^2(\Z)$. Remark \ref{extunit} shows we can nevertheless apply our results to this case.

Setting $\Lambda=\{1,2,\dots, 2n\}$, we define $T^\Lambda=D_\omega^\Lambda T^\Lambda$, with $D_\omega^\Lambda$ unitary and diagonal and 
\be
T^\Lambda=\begin{pmatrix}
0 & 0 & g & & & & \cr
1 & 0 & 0 & & & &\cr
   & 0 & 0 & & & &\cr               
   & 1&0 & \ddots & 0& g & \cr
   & & & & 0 & 0 &  \cr
   & & & & 0 & 0 & g\cr
   & & & & 1 &0 & 0
\end{pmatrix}.
\ee
This yields a finite volume restriction of $T_\omega$, that is a cnu contraction with eigenvalues $\{\lambda_j\}_{j=1}^{2n}$
\be
\lambda_j(\omega)=\sqrt g e^{\frac{i}{2n}\sum_{k=1}^{2n}{\omega_k}}e^{i(j-1)\pi/n}, \ \ j=1,\dots,2n.
\ee
The factor $\frac{1}{2n}\sum_{k=1}^{2n}{\omega_k}$ tends to $0$ as $n\ra\infty$ almost surely, by our assumption on the distribution of the phases, and we get for any $f\in C(\bar \D)$,
\be
\lim_{n\ra\infty}\frac{1}{2n}\sum_{j=1}^{2n} f(\Re\lambda_j(\omega), \Im\lambda_j(\omega))= \int_{[0,2\pi]}f(\sqrt g \cos(\theta), \sqrt g \sin(\theta))\frac{d\theta}{2\pi},
\ee
by a Riemann sum argument. In other words, the normalised counting measure on $\sigma(T^{\Lambda}_\omega)$ converges weakly to $dm=\delta_{\sqrt g}(r)\frac{d\theta}{2\pi}$, in polar coordinates. As $\spr (T^\Lambda_\omega)=\sqrt g<1$, we compute from Theorem \ref{convcount},
\be
 \psi(z)=\frac{1}{2\pi}\int_{[0,2\pi]}\frac{d\theta}{1-z\sqrt g e^{-i\theta}}=1, \ \ \forall z\ \mbox{s.t.} \ |z|<1/\sqrt g,
\ee
so that we recover $\ffi(e^{it})=1$, in keeping with Remark \ref{xxx}
\begin{rem}
The support of the limiting measure $dm$ is disjoint from $\sigma(T_\omega)$ in this case:
\be
\mbox{supp }dm\cap \sigma(T_\omega)=\sqrt g\partial \D\cap \{g\partial \D\cup\partial D\}=\emptyset. 
\ee
\end{rem}
 {

}

\begin{thebibliography}{99}
%
\bibitem[BHJ]{bhj} O. Bourget, J. S. Howland and A. Joye,
{Spectral Analysis of Unitary Band Matrices}.{\it Commun. Math. Phys.} 
{\bf 234}, (2003), 191-227.
%
\bibitem[CL]{cl} R. Carmona, J. Lacroix,  {\it Spectral theory of random
    Schrodinger Operators}, Birkhauser, 1990.
%
\bibitem[CFKS]{cfks} H.L. Cycon , R.G. Froese, W. Kirsch, B. Simon,
    {\it Schr\"odinger Operators}, Springer Verlag, 1987.
%
\bibitem[D1]{D1} E. B. Davies, Spectral theory of pseudo-ergodic operators. {\it Comm. Math. Phys.} {\bf 216} (2001), 687-704. 
%
\bibitem[D2]{D2} E. B. Davies, Spectral properties of random non-self-adjoint matrices and operators. {\it Proc. R. Soc. Lond. A.} {\bf 457} (2001) 191-206.
%
\bibitem[G]{G} M. Garnett, {\it Bounded Analytic Functions}, Graduate Texts in lathematics 236, Springer, 2007.
%
\bibitem[GoKh1]{GK} I. Y. Goldsheid and B. A. Khoruzhenko, Distribution of eigenvalues in non-Hermitian Anderson models. {\it Phys. Rev. Letters}, {\bf 80}(13) (1998), 2897. 
%
\bibitem[GoKh2]{GK2} I. Y. Goldsheid and B. A. Khoruzhenko,  Eigenvalue curves of asymmetric tridiagonal random matrices. {\it Electronic Journal of Probability}, {\bf 5}, 1-28, (2000).
%
\bibitem[HJ]{HJ2} E.\ Hamza and  A.\ Joye,  Spectral Properties of Non-Unitary Band Matrices.
 {\it Ann. Henri Poincar\'e}, to appear.
%
\bibitem[HJS]{HJS} E.\ Hamza, A.\ Joye and G.\ Stolz, Dynamical Localization for Unitary Anderson Models. 
{\it Math. Phys., Anal. Geom.} {\bf 12} (2009), 381-444.
%
\bibitem[J1]{J1} A. \ Joye,  Density of States and Thouless Formula for Random Unitary Band Matrices. {\it Ann. H. Poincar\'e}. {\bf 5} (2004), 347-379.
%
\bibitem[Ki]{Ki} W.\ Kirsch, An invitation to random Schr\"odinger operators, Random Schr\"odinger Operators, {\it Panorama et Synth\`ese}, Soc. Math. France, Vol. 25, 1-119, 2008
%
\bibitem[R]{R} W. Rudin, {\it Real and Complex Analysis}, McGraw-Hill, 1976.
%
\bibitem[SFBK]{SNF} B. Sz.-Nagy, C. Foias, H. Berkovici, L. K\'erchy: Harmonic Analysis of Operators in Hilbert Spaces. Springer (2010).
%
\bibitem[TE]{TE} L.\ Trefethen, M.\ Embree, Spectra and Pseudospectra: The Behaviour of Nonnormal Matrices and Operators, PUP, (2005).
%
\end{thebibliography}
\end{document}